\documentclass[12pt]{article}

\usepackage{amsfonts, amsmath, amssymb, amsthm}
\usepackage{a4}

\usepackage{graphicx}

\newtheorem{thm}{Theorem}
\newtheorem{lemma}{Lemma}

\theoremstyle{definition}
\newtheorem{dfn}{Definition}
\newtheorem{Ex}{Example}

\theoremstyle{remark}
\newtheorem*{rem}{Remark}

\newcommand{\Z}{} \def\Z{{\mathbb Z}}

\DeclareMathOperator{\ord}{ord}
\DeclareMathOperator{\End}{End}

\title{Cohomologies of $n$-simplex relations}
\author{I.\,G. Korepanov, G.\,I. Sharygin and D.\,V. Talalaev}

\date{}

\begin{document}

\maketitle

\begin{abstract}
A theory of (co)homologies related to set-theoretic $n$-simplex relations is constructed in analogy with the known quandle and Yang--Baxter (co)ho\-mologies, with emphasis made on the tetrahedron case. In particular, this permits us to generalize Hietarinta's idea of ``permutation-type'' solutions to the quantum (or ``tensor'') $n$-simplex equations. Explicit examples of solutions to the tetrahedron equation involving nontrivial cocycles are presented.
\end{abstract}

\tableofcontents

\section{Introduction}

$n$-Simplex \emph{equations} arose as fundamental equations underlying exactly solvable models in mathematical physics. First, there appeared the 2-simplex, or \emph{triangle}, equation, which is nothing but the famous Yang--Baxter equation (YBE)~\cite{YBE1,YBE2}, and the second to appear was the 3-simplex, or \emph{Zamolodchikov tetrahedron} equation (ZTE)~\cite{ZTE}. Loosely speaking, $n$-simplex equations are those representable graphically with one $n$-simplex in its left-hand side, and the same $n$-simplex but turned ``inside out'' in its right-hand side. To be more exact, there are quite a few different types of $n$-simplex equations, of which important for us here will be:
\begin{itemize}\itemsep 0pt
 \item \emph{quantum}, or \emph{tensor}, equation, or equation relating some products of linear operators acting in the tensor product of vector spaces. See formula~\eqref{tetra} below for a quantum \emph{tetrahedron} equation,
 \item \emph{set-theoretic} equation, or equation relating the compositions of mappings of some (usually finite) sets. These compositions act in a suitable Cartesian product of the mentioned sets. Set-theoretic Yang--Baxter equation was introduced in~\cite{Dri},
 \item \emph{functional} equation. This is almost the same as set-theoretic one, except that the mappings are not required to be defined everywhere. Typically, they are given by \emph{rational functions}. The functional \emph{tetrahedron} equation (FTE), together with a rather general scheme giving its solutions, was proposed in~\cite{FTE}.
\end{itemize}

Different kinds of $n$-simplex equations prove to be closely related to each other. For instance, FTE appears to provide the most fruitful way for finding solutions of the quantum tetrahedron equation (QTE). The first example can be found already in~\cite{FTE}; also, the example~\cite{S-et-al} with \emph{positive Boltzmann weights} (so important for statistical physics) was constructed using the FTE.

\begin{rem}
Many more types of $n$-simplex equations and their classical analogues can be found in literature. For instance, we can mention the \emph{classical} Yang--Baxter equation used in the theory of (classical) solitonic equations. There are also equations in the \emph{direct sum} (rather than tensor product) of linear spaces, see again~\cite{FTE} for a Yang--Baxter direct-sum example, which is also \emph{dynamical}: the operators in its l.h.s.\ and r.h.s.\ can be different (in contrast with the ``usual'' equations, where the operators/mappings are the same, but taken in different orders). Also, our quantum, set-theoretic and functional equations correspond to coloring edges around a vertex or, in the dual picture, $(n-1)$-faces around an $n$-cube (see Section~\ref{s:n-simplex} below), while there are also options of coloring faces of other dimensions. Note in this connection that colors in Zamolodchikov's original paper~\cite{ZTE} were attached to \emph{two-faces}.
\end{rem}

An $n$-simplex equation together with operators/mappings making its solution is called $n$-simplex \emph{relation}.

Besides mathematical physics, $n$-simplex equations/relations look very important for \emph{topology}. Here it is quite enough to mention the connections between the Jones polynomial and Yang--Baxter equation, see, for instance, \cite{Turaev} and references therein.

One direction of searching for new solutions to the $n$-simplex equations is suggested, on one hand, by Hietarinta's \emph{permutation-type} solutions~\cite{Hie} and, on another hand, by the \emph{quandle cohomology}~\cite{JSC} and \emph{Yang--Baxter (co)homology}~\cite{CES}. 

The aim of the present paper is to introduce a similar (co)homological theory for general $n$-simplex relations and show how it leads to nontrivial generalizations of permutation-type solutions. Although we develop a general theory, our main attention is paid to the tetrahedron case $n=3$. Also, the examples of permutation-type solutions used in this paper are our own.

The organization of the rest of the paper is as follows:
\begin{itemize}\itemsep 0pt
 \item in Section~\ref{s:tetra}, we introduce various forms of tetrahedron equation,
 \item in Section~\ref{s:n-simplex}, we generalize the set-theoretic version to any $n$-simplex equation, and study the related ``permitted'' colorings of a ``big-dimensional'' cube,
 \item in Section~\ref{s:hom}, we define homologies and cohomologies of $n$-simplex relations, based on the mentioned ``permitted'' colorings,
 \item and finally, in Section~\ref{s:examples}, we show how our theory can give nontrivial solutions to a quantum (tensor) tetrahedron equation.
\end{itemize}

\section{Tetrahedron equations}\label{s:tetra}

\subsection{Quantum, set-theoretic and functional tetrahedron equations}\label{ss:TE}

The tetrahedron equation~(TE) arose initially (in one of its versions) in Zamo\-lod\-chi\-kov's paper~\cite{ZTE} where he considered the scattering of straight strings in $(2+1)$ dimensions. As TE is a higher-dimen\-sional analogue of Yang--Baxter equation, it is believed to find, similarly to YBE, its applications to solvable models in statistical physics and quantum field theory; also applications to topology are likely to be developed.

\paragraph{Quantum (tensor) equation.}
We write out here, first, the quantum constant version of TE:
\begin{equation}\label{tetra}
\Phi_{123}\Phi_{145}\Phi_{246}\Phi_{356}=\Phi_{356}\Phi_{246}\Phi_{145}\Phi_{123}.
\end{equation}
Here both sides are operators in the $6$th tensor degree of a linear space~$V^{\otimes 6}$; we write it as
\[
V^{\otimes 6}=V_1 \otimes V_2 \otimes V_3 \otimes V_4 \otimes V_5 \otimes V_6,
\]
where $V_1, \dots, V_6$ are six copies of~$V$; each separate~$\Phi_{ijk}$ is a copy of one and the same linear operator $\Phi\in \End(V^{\otimes 3})$ (hence the name ``constant equation''), acting nontrivially in $V_i \otimes V_j \otimes V_k$, and tensor multiplied by identity operators in the remaining~$V$'s.

Equation~\eqref{tetra} is represented graphically in Figure~\ref{f:t}.
\begin{figure}[htp]
 \begin{center}
  \includegraphics[scale=1.4]{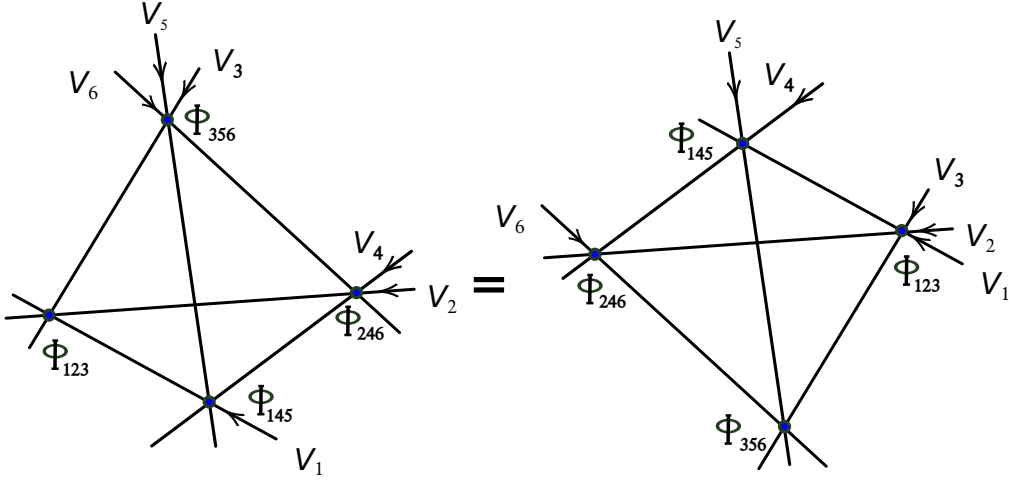}
 \end{center}
\caption{Graphical representation of the tetrahedron equation. Lines correspond to vector spaces, dots --- to operators~$\Phi_{ijk}$.}
\label{f:t}
\end{figure}

\begin{rem}
If the space~$V$ is finite-dimen\-sional and has a chosen basis --- which is very often the case --- then equation~\eqref{tetra} can be written in the ``coordinate'' form. Each operator~$\Phi_{abc}$ is represented as a tensor with three superscripts and three subscripts, and the equation reads
\begin{equation}\label{tetra-coord}
\Phi_{i_1i_2i_3}^{j_1j_2j_3} \Phi_{j_1i_4i_5}^{k_1j_4j_5} \Phi_{j_2j_4i_6}^{k_2k_4j_6} \Phi_{j_3j_5j_6}^{k_3k_5k_6} = \Phi_{i_3i_5i_6}^{j_3j_5j_6} \Phi_{i_2i_4j_6}^{j_2j_4k_6} \Phi_{i_1j_4j_5}^{j_1k_4k_5} \Phi_{j_1j_2j_3}^{k_1k_2k_3}.
\end{equation}
In~\eqref{tetra-coord}, the summation is implied over each pair of coinciding subscript and superscript.
\end{rem}

\paragraph{Set-theoretic equation.}
Our next version of TE is the \emph{set-theoretic tetrahedron equation} (STTE; compare with set-theoretic YBE~\cite{Dri,CES}). It is defined as follows: let $X$ be a set; then an STTE on~$X$ is
\begin{equation}
\label{eqtetr1}
R_{123} \circ R_{145} \circ R_{246} \circ R_{356} = R_{356} \circ R_{246} \circ R_{145} \circ R_{123}\;\colon\quad X^{\times 6}\to X^{\times 6},
\end{equation}
where $R_{ijk}$ are copies of a map
\[
X\times X\times X\stackrel{R}{\longrightarrow} X\times X\times X.
\]
Note that \eqref{eqtetr1} is represented graphically by the same Figure~\ref{f:t} (with $\Phi$ changed to~$R$ everywhere), but now we are dealing with the sixth \emph{Cartesian}, rather than tensor, degree~$X^{\times 6}$ of the \emph{set}~$X$. Of course, the subscripts indicate again the numbers of copies of~$X$ where $R$ is applied nontrivially, while it acts as the identity map on the other copies, for instance,
\begin{equation}\label{eqtetr}
\begin{aligned}
R_{356}(a_1,a_2,a_3,a_4,a_5,a_6)&=\bigl( a_1,a_2,R_1(a_3,a_5,a_6),a_4,R_2(a_3,a_5,a_6),R_3(a_3,a_5,a_6)\bigr)\\&=(a_1,a_2,a_3',a_4,a_5',a_6'),
\end{aligned}
\end{equation}
where
\begin{equation}\label{R}
R(x,y,z)=\bigl(R_1(x,y,z),R_2(x,y,z),R_3(x,y,z)\bigr)=(x',y',z').
\end{equation}

Clearly, if we consider a vector space~$V$ whose basis is our set~$X$, then the map~$R$ will generate, in an obvious way, the linear operator $\Phi=\Phi_R$ satisfying~\eqref{tetra}. Such~$\Phi$ is called a \emph{permutation-type} operator, compare~\cite{Hie}. We will consider such operators in more detail in Subsection~\ref{ss:QTE-cocycles}. On the other hand, there exist much more general solutions of QTE that cannot be reduced to an STTE.

\paragraph{Functional equation.}
Our third type of tetrahedron equation is FTE --- the functional tetrahedron equation. This is the same as STTE except that the mapping~$R$ may be defined ``almost everywhere'' rather than be a mapping between sets in the strict sense. Typically, $X$ is in this situation a field like $\mathbb R$ or~$\mathbb C$, and $R$ is a rational function. One popular solution to FTE is the following ``electric'' solution:
\begin{equation}\label{el}
\begin{aligned}
R(x,y,z)&=(x',y',z');\\
x'&=\frac {x y} {x+z+x y z},\\
y'&=x+z+x y z,\\
z'&=\frac{ yz} {x+z+x y z}.
\end{aligned}
\end{equation}
We call it ``electric'' because it comes from the well known \emph{star--triangle} transformation in electric circuits, see~\cite{FTE}.

\begin{rem}
We will show in Subsection~\ref{ss:sol} how a genuine bijective map between sets can be fabricated from the rational mapping~\eqref{el}.
\end{rem}

\subsection{Interpretation of tetrahedron equation in terms of coloring of cube faces}\label{ss:tf}

One more way to represent equation~\eqref{eqtetr} graphically is \emph{dual} to Figure~\ref{f:t} and goes as follows. Consider a $4$-dimensional cube. Its orthogonal projection onto a three-dimen\-sional space along its (body) diagonal is a rhombic dodecahedron, wherein its eight three-dimen\-sional faces become eight parallelepipeds. In this way, our rhombic dodecahedron turns out to be divided into four parallelepipeds in two different ways, see Figure~\ref{f:r}, where the indices at vertices indicate the vertex coordinates in~$\mathbb R^4$.
\begin{figure}[htp]
 \begin{center}
 \includegraphics[scale=1.3]{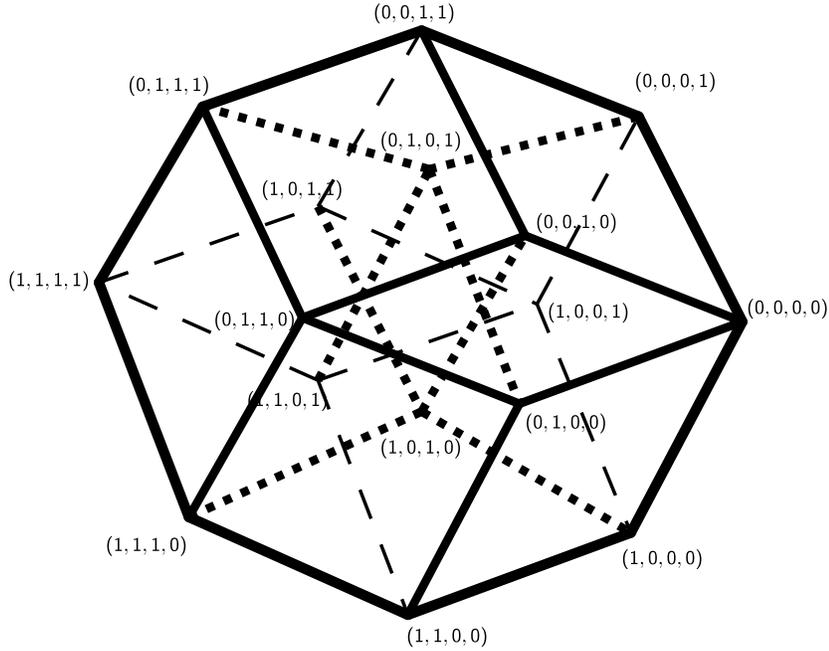}
 \end{center}
 \caption{Rhombic dodecahedron, split into parallelepipeds}
 \label{f:r}
\end{figure}

We now position our rhombic dodecahedron in the three-space in such way that it will have $6$ visible and $6$ invisible 2-faces.
Paint the visible faces in $6$ colors $a_1,\dots,a_6$ from the set~$X$ as shown in the upper part of Figure~\ref{f:f}
\begin{figure}[htp]
 \begin{center}
\includegraphics[scale=.8]{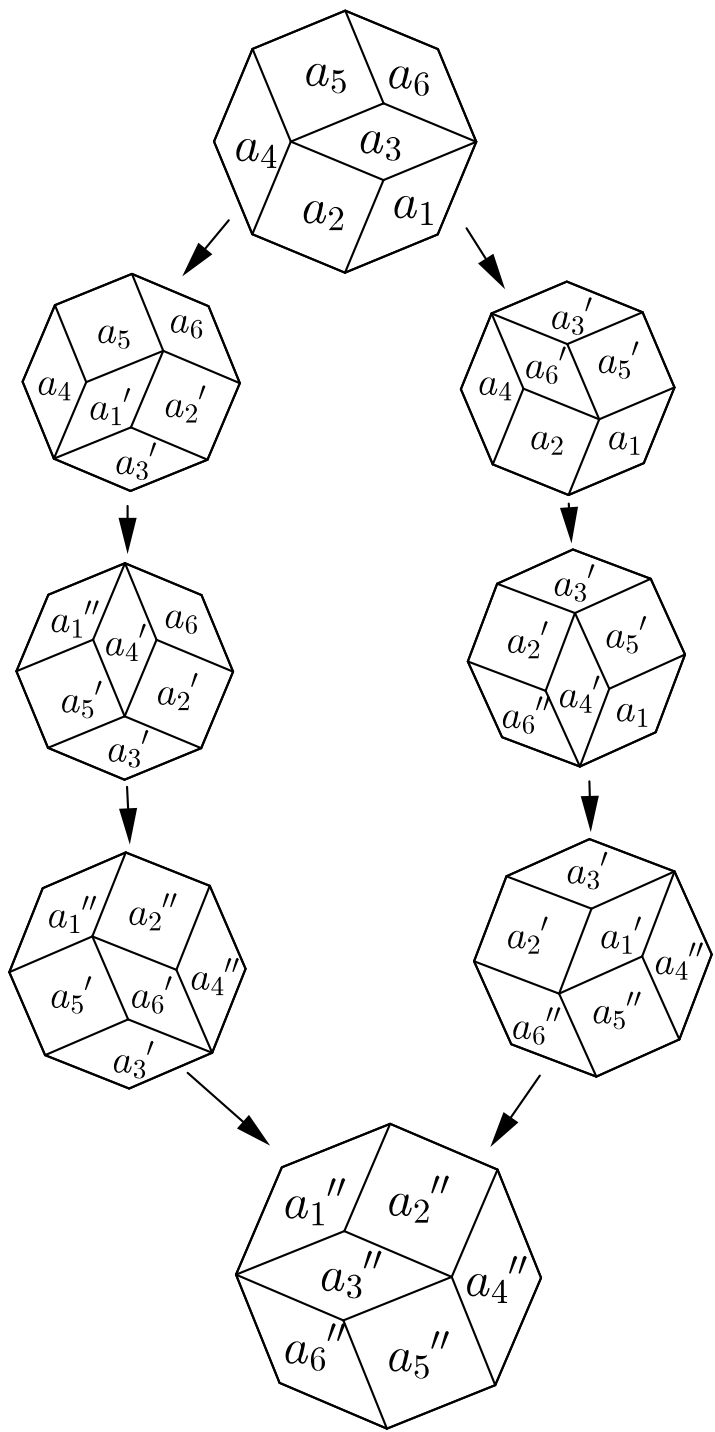}
 \end{center}
 \caption{One more graphical representation of the tetrahedron equation --- dual to Figure~\ref{f:t}}
 \label{f:f}
\end{figure}
(where a flat projection of rhombic dodecahedron is pictured, in the shape of regular octagon).

Consider one of two partitionings of the rhombic dodecahedron into parallelepipeds, depicted in Figure~\ref{f:r}. As we see, the colors of three front (visible) faces of one parallelepiped in this partitioning are $a_1,a_2,a_3$. Color then its three remaining faces with colors $a_1',a_2',a_3'$ (obtained from $a_1,a_2,a_3$ by applying the mapping~$R$, compare~\eqref{R}) respectively. Now there has appeared a parallelepiped whose three visible faces are colored with $a_1',a_4,a_5$, while its three invisible faces are not yet colored. We color them in $(a_1'',a_4',a_5')=R(a_1',a_4,a_5)$, and so on.

This process is represented in the left column in Figure~\ref{f:f} (the already used faces are not pictured). We could, however, begin the process from the parallelepiped whose visible faces have colors $a_3,a_5,a_6$, and thus obtain the sequence of transformations in the right column. It is not hard to see that the tetrahedron equation states exactly that the final result --- the coloring of the invisible part of the rhombic dodecahedron surface (shown below in Figure~\ref{f:f}) --- is the same in both cases. Besides, if we consider the coloring of all the involved 2-faces, including those which we used in passing, we get the coloring of all 2-faces of the tesseract (i.e., four-dimen\-sional cube) meeting the condition that, for any 3-face, the coloring of its three ``back'' faces is obtained by applying the mapping~$R$ onto the colors of three ``front'' faces.

Note also that we could write, instead of our equation~\eqref{eqtetr1}, any other equation obtained from it using index permutations. Indeed, the group~$S_6$ acts on the Cartesian product~$X^{\times 6}$ by permutations and, clearly, the conjugation by an element $\sigma\in S_6$ preserves the structure of equation~\eqref{eqtetr1}. We call below any such equation tetrahedron equation.

\section{General $n$-simplex relations and colorings of $N$-cube faces}\label{s:n-simplex}

Both Yang--Baxter and tetrahedron relations are particular cases of their natural generalization --- the $n$-simplex relation. In this Section, we first do some preparatory work in Subsection~\ref{ss:0}. Then we define the $n$-simplex relation, in its set-theoretic version, in Subsection~\ref{ss:ur}. This goes in the style of Subsection~\ref{ss:tf} --- using face colorings of a cube. Finally, we investigate some specific colorings of $(n-1)$-faces in a cube of ``big'' dimension~$N$ in Subsection~\ref{ss:po}.

\subsection{Cube in space~$\mathbb R^N$ and its faces}\label{ss:0}

Consider the cube~$I^N$ in the $N$-dimensional space $\mathbb R^N\ni (x_1,\dots,x_N)$. By definition, $I^N$ is given by the following system of inequalities:
\[
\begin{cases}
0\le x_1\le 1,\\
\dotfill \\
0\le x_N\le 1 .
\end{cases}
\]
Any face~$f_k\subset I^N$ of dimension~$k$, or simply $k$-face, $0\le k\le N$, will be specified, in the spirit of paper~\cite{MR}, by a sequence~$\tau_k$ of $N$ symbols $0$, $1$ or~$*$ (in~\cite{MR}, the symbol~$x$ was used instead of our asterisk). This means, by definition, that~$f_k$ is determined by the following system of equalities and inequalities in~$\mathbb R^N$:
\[
\begin{cases}
x_i=0 & \text{for all }i\text{ such that the }i\text{-th term is }0,\\
x_i=1 & \text{for all }i\text{ such that the }i\text{-th term is }1,\\
0\le x_i\le 1 & \text{for all }i\text{ such that the }i\text{-th term is }*.
\end{cases}
\]
Clearly, the total number of asterisks in such sequence is~$k$. For instance, the origin of coordinates is given by the sequence $(0\;0\;0\ldots0)$, while the whole cube~$I^N$ --- by the sequence $(*\;*\;*\;\ldots\;*)$.

We will consider the \emph{set-theoretic $n$-simplex equation}, $2\le n<N$. For this equation, faces of dimensions $n-1$, $n$ and~$n+1$ will be of special importance. Let an $n$-face~$f_n$ be given by a sequence~$\tau_n$ of $N$ elements $0$, $1$ or~$*$, as described above. Any $(n-1)$-subface $g_{n-1}\subset f_n$ is determined by the same sequence~$\tau_n$, but with one more asterisk replaced by a digit, i.e., with the equation
\begin{equation}\label{g}
x_{j_k}=0 \text{ or } 1
\end{equation}
in place of the corresponding inequality.

Denote the numbers of positions of asterisks in sequence~$\tau_n$ as $j_1<\dots<j_n$. Then, put alternating zeros and unities in correspondence to these numbers, that is, introduce the quantities
\[
\varkappa_{j_1}=0, \quad \varkappa_{j_2}=1, \quad \varkappa_{j_3}=0, \quad \text{etc.}
\]

\begin{dfn}\label{dfn:io}
If the r.h.s.\ of equation~\eqref{g} coincides with~$\varkappa_{j_k}$, then the face~$g_{n-1}$ is called \emph{incoming} for~$f_n$; if not, it is called \emph{outgoing}.
\end{dfn}

The motivation for this definition will become clear when we introduce the set-theoretic $n$-simplex equation in Subsection~\ref{ss:ur}. Note that any $n$-face of cube $I^N$,\; $N>n$, has $2n$ subfaces of dimension~$(n-1)$, of which $n$ are incoming and $n$ are outgoing.

\begin{dfn}\label{dfn:o}
For an $(n-1)$-face $g_{n-1}\subset I_N$, its \emph{order} $\ord g_{n-1}$ is the number of $n$-faces~$f_n$ containing it, $g_{n-1}\subset f_n \subset I_N$, and such that $g_{n-1}$ is incoming for~$f_n$.
\end{dfn}

Next comes one more important definition:

\begin{dfn}\label{dfn:abs}
An $(n-1)$-face is called \emph{absolutely incoming} if it has the maximal possible order $N-n+1$, i.e., it is incoming for all $n$-faces containing it. An $(n-1)$-face of order~$0$ is called \emph{absolutely outgoing}.
\end{dfn}

\begin{lemma}\label{l:abs}
The $N$-cube has exactly $\binom{N}{n-1}$ absolutely incoming, as well as $\binom{N}{n-1}$ absolutely outgoing, $(n-1)$-faces.
\end{lemma}

\begin{proof}
We call any collection of $n-1$ basis vectors in~$\mathbb R^N$ \emph{$(n-1)$-direction}. Clearly, any $(n-1)$-face in $N$-cube is parallel to exactly one $(n-1)$-direction; there are $\binom{N}{n-1}$ such directions, and there are $2^{N-n+1}$ faces parallel to a given $(n-1)$-direction. 

We are going to show that, for any $(n-1)$-direction, there exists exactly one absolutely incoming (or absolutely outgoing) $(n-1)$-face parallel to it. Indeed, a face is parallel to a direction $e_{i_1},\dots,e_{i_{n-1}}$ provided the asterisks~$\ast$ in the corresponding sequence~$\tau_{n-1}$ stand at the positions $i_1,\dots,i_{n-1}$ (and only there). Any $n$-face containing a given $(n-1)$-face is characterized by the fact that the corresponding sequence~$\tau_n$ has, in addition to the asterisks in~$\tau_{n-1}$ (sequence for the mentioned $(n-1)$-face), one more asterisk that can stand at any of the remaining $N-n+1$ positions. The condition that a given $(n-1)$-face is incoming (or outgoing) for an $n$-face containing it, determines unambiguously the symbol standing in~$\tau_{n-1}$ at the position of the ``extra'' asterisk in~$\tau_n$. Thus, the sequence~$\tau_{n-1}$ is restored unambiguously from the set of indices $i_1,\dots,i_{n-1}$ and the condition of being absolutely incoming (or outgoing).
\end{proof}

\subsection{The $n$-simplex equation and its geometrical meaning}\label{ss:ur}

We are going to give a ``geometrical'' definition of $n$-simplex equation, in analogy with that for tetrahedron equation in Subsection~\ref{ss:tf}. We start with the definition of \emph{coloring} of the $N$-cube.

\begin{dfn}\label{dfn:clr}
\emph{Coloring} of $(n-1)$-faces in cube~$I^N$ is a mapping from the set of these faces into a fixed finite set~$X$ --- \emph{color set}. We call a coloring of $(n-1)$-faces $(n-1)$-coloring or even simply coloring, if the value of~$n$ is clear from the context.
\end{dfn}

Now choose an arbitrary $n$-face~$f_n\subset I^N$; it contains $n$ incoming and $n$ outgoing $(n-1)$-faces (see Definition~\ref{dfn:io} and the following Remark). So, a coloring of either all incoming or all outgoing faces is given by an element of the Cartesian degree~$X^{\times n}$.

\begin{dfn}\label{dfn:R}
\emph{Set-theoretic $R$-operator} is a mapping
\[
R\colon\; X^{\times n} \to X^{\times n}
\]
producing the colors of outgoing $(n-1)$-faces in~$f_n$ if the colors of incoming faces are given.
\end{dfn}

With such an operator, we define special colorings of a ``big'' cube:

\begin{dfn}\label{dfn:prm}
\emph{Permitted coloring} of $(n-1)$-faces in the cube~$I^N$ with respect to a given $R$-operator is a coloring where the colors of outgoing faces in any $n$-(sub)cube~$f_n\subset I^N$ are obtained from the ``incoming'' colors using the operator~$R$.
\end{dfn}

For the set of permitted colorings to be non-empty, conditions on different $n$-faces must not contradict each other. It turns out that just one condition on the $R$-operator is enough for this, and this is exactly what we will call the set-theoretical $n$-simplex equation. 

To introduce this equation, consider an $(n+1)$-cube~$I^{n+1}$.  Build the following directed graph, which we will call~$G_{n+1}$: first, take all $n$-faces of~$I^{n+1}$ as its \emph{vertices}. Then, consider those $(n-1)$-faces that are outgoing for one $n$-face and incoming for another, and take them as the graph~$G_{n+1}$ \emph{edges}, joining the two mentioned $n$-faces and directed from the first of them to the second. It turns out that~$G_{n+1}$ consists of two $n$-simplices, one of them ``turned inside out'' with respect to the other (compare Figure~\ref{f:t}), as the following theorem states. We continue to use the notations introduced in Subsection~\ref{ss:0}, writing the faces of~$I^{n+1}$ as sequences of $n+1$ symbols $0$, $1$ or~$*$. 

\begin{thm}\label{thm:n-simp}
Graph~$G_{n+1}$ consists of two connected components, each being an $n$-simplex. To one of them --- we will call it the \emph{left-hand-side} component --- belong the $n$-faces (graph vertices) having zero at an odd position or unity at an even position:
\begin{equation}
\label{left}
(0***\ldots*),\quad ({}*1**\ldots*),\quad (**0*\ldots*),\quad \ldots \;.
\end{equation}
The rest of $n$-faces, i.e.,
\begin{equation}
\label{right}
(1***\ldots*),\quad ({}*0**\ldots*),\quad (**1*\ldots*),\quad \ldots \;,
\end{equation}
belong to the other component, called the \emph{right-hand-side} component.

Moreover, let us say that $a<b$ for two vertices $a,b\in G_{n+1}$ if one can get from $a$ to~$b$ along (a sequence of) \emph{directed} edges. Then $<$ makes a linear order on any of the connected components, namely,
\begin{equation}\label{l<}
(0***\ldots*) \,<\, ({}*1**\ldots*) \,<\, (**0*\ldots*) \,<\, \ldots
\end{equation}
in the l.h.s., and
\begin{equation}\label{r>}
(1***\ldots*) \,>\, ({}*0**\ldots*) \,>\, (**1*\ldots*) \,>\, \ldots
\end{equation}
in the r.h.s.
\end{thm}

The simplest example --- graph~$G_3$, consisting of two triangles --- is pictured in Figure~\ref{f:c}.
\begin{figure}[ht]
 \begin{center}
  \includegraphics[scale=1.5]{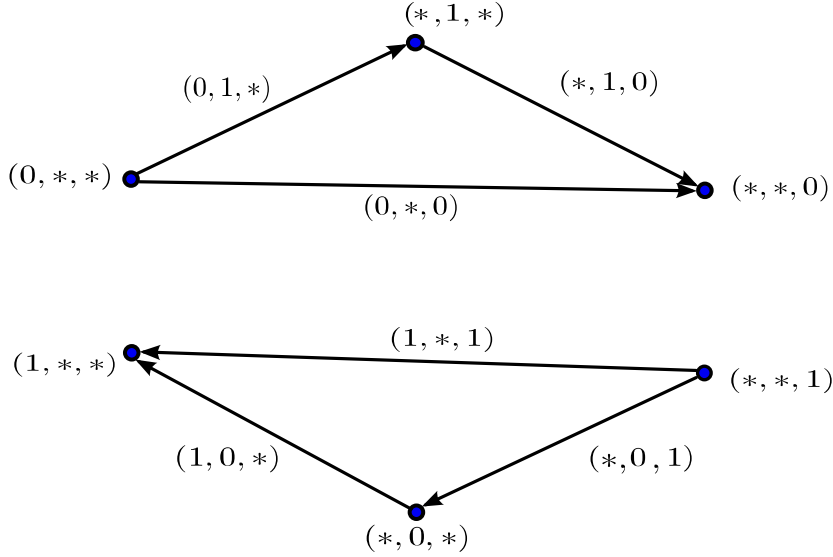}
  \caption{Incoming and outgoing edges in cube~$I^3$. Sequences encoding edges are written near the corresponding arrows; graph vertices denote 2-faces.}
  \label{f:c}
 \end{center}
\end{figure}

\begin{proof}
Consider the set~\eqref{left} of graph~$G_{n+1}$ vertices. There are clearly $n+1$ of them. First, we show that, for any two of them, call them $a$ and~$b$, if $b$ is to the right of~$a$ in~\eqref{left} (although not necessarily its nearest neighbor!), then there exists an edge joining $a$ and~$b$ directly and pointing from $a$ to~$b$. Indeed, one can check that the \emph{intersection of $n$-faces} $a$ and~$b$ is exactly an $(n-1)$-face outgoing for~$a$ and incoming for~$b$ --- this becomes an easy exercise if we note that the sequences in~\eqref{left} and their $(n-1)$-faces are obtained from the sequence $(0\;1\;0\;1\;\dots)$ by replacing all the symbols, except one or two, with~$\ast$, and then directly apply Definition~\ref{dfn:io}.

So, restricting the set of vertices to those in~\eqref{left} and considering only edges joining these vertices, we have an $n$-simplex, as desired. Similar reasoning is valid, of course, also for the set~\eqref{right}.

What remains to prove is that there are no more edges (that would join the two already obtained $n$-simplices). Indeed, taking now $n$-face~$a$ from~\eqref{left} and $n$-face~$b$ from~\eqref{right} (obtained from the sequence $(0\;1\;0\;1\;\dots)$ or, respectively, $(1\;0\;1\;0\;\dots)$ by replacing all the symbols, except one, with~$\ast$) and applying again Definition~\ref{dfn:io}, one can check that $a\cap b$ is either incoming or outgoing for \emph{both} $a$ and~$b$.
\end{proof}

\begin{dfn}\label{dfn:n-simp}
\emph{Set-theoretic (constant) $n$-simplex equation} on a color set~$X$ is the following equality between two compositions of $R$-operators, acting in the increasing order (from right to left!) in the sense of \eqref{l<} and~\eqref{r>}:
\begin{equation}\label{eq:n-simp}
\cdots \circ R_{(**0*\ldots*)} \circ R_{({}*1**\ldots*)} \circ R_{(0***\ldots*)} =
 R_{(1***\ldots*)} \circ R_{({}*0**\ldots*)} \circ R_{(**1*\ldots*)} \circ \cdots \,.
\end{equation}
Each $R$ is a copy of the same $R$-operator but acting on the colors within the $n$-face in its subscript.
\end{dfn}

\begin{rem}
Thus, the l.h.s.\ of the set-theoretic $n$-simplex equation corresponds graphically to the left-hand-side component of~$G_{n+1}$, and the r.h.s.\ similarly to the right-hand-side component.
\end{rem}

The following geometric interpretation can also be given to the definitions in this Subsection.

We define the \emph{orientation} of any $k$-face in cube~$I^N$, assuming that its edges go along the axes $x_{j_1},\dots,x_{j_k}$, with $j_1<\dots<j_k$, as follows: it is determined by the $k$-tuple $\mathsf e_{j_1},\dots,\mathsf e_{j_k}$ of unit vectors going along these axes and taken exactly in this order of increasing indices. To be more exact, we identify $\mathsf e_{j_1},\dots,\mathsf e_{j_k}$ with the standard basis $\mathsf f_{1},\dots,\mathsf f_{k}$ in~$\mathbb R^k$, and thus carry over the standard orientation of~$\mathbb R^k$ onto the face. Any subface $I^l\subseteq I^k\subseteq I^N$ gets thus two orientations: one from~$I^N$ and one from~$I^k$; these orientations, of course, coincide.

As $\mathbb R^N$ has its standard orientation as well, choosing face orientation is equivalent to choosing an orientation in the complementary space. Below we assume that we compose the basis in~$\mathbb R^N$, taking first the vectors $\mathsf e_{j_1},\dots,\mathsf e_{j_k}$, and then the basis in the complementary space; this must give the correct orientation of~$\mathbb R^N$. In particular, if $k=N-1$, then the face orientation is equivalent to choosing a direction of the normal to the face. As the orientations of opposite faces clearly coincide, we see that one half of normals to the $n$-faces of $(n+1)$-cube look into that cube, while the other half look outside. We then say that faces looking inside correspond, by definition, to the l.h.s.\ of our equation, while those looking outside --- to the r.h.s.

Then, we define, for every $n$-face (recall that it must correspond to one $R$-operator) \emph{its incoming and outgoing} $(n-1)$-faces. We do it along the same lines as in the previous paragraph: those whose normals (drawn within the $n$-face under consideration) point inside are incoming, while the rest are outgoing. Important is that the orientations of all $n$-faces in either side of our equation are obviously \emph{consistent}! And this means that if an $(n-1)$-face is common for two $n$-faces, then it is incoming for one of them and outgoing to the other. It is not hard to check that this way we can come to the definition of $n$-simplex equation equivalent to the above.

\subsection{Partial order on $n$- and $(n-1)$-faces}\label{ss:po}

The main result of this Subsection is the following theorem, forming the basis for homological constructions in Section~\ref{s:hom}.

\begin{thm}\label{thm:vse}
Assume that $X$ is a set, equipped with an operation $R\colon \; X^{\times n}\to X^{\times n}$, for which the $n$-simplex equation holds. Then for every $X$-coloring of absolutely incoming $(n-1)$-faces of $I^N$ there is a unique permitted $X$-coloring of all $(n-1)$-faces of the $N$-cube.
\end{thm}

\begin{proof}
Consider the following directed (i.e every edge bears ``an arrow'') graph $\Gamma_N$: its vertices correspond to $n$\/-dimensional and $(n-1)$\/-dimensional faces of the cube $I^N$ (we shall denote them $f_n$ and $f_{n-1}$ respectively). An edge of $\Gamma_N$ connects an $n$-face $f_n$ with an $(n-1)$-face $f_{n-1}$ iff $f_{n-1}\subset f_n$; it is directed from $f_n$, if $f_{n-1}$ is an outgoing face for $f_n$, and it is directed towards $f_n$ if $f_{n-1}$ is an incoming face; there are no other edges. We shall prove two important properties of this graph: 

\begin{lemma}\label{l:acycl}
There are no oriented cycles in~$\Gamma_N$ (loosely speaking, you cannot return to the starting point, if you follow the arrows of the graph).
\end{lemma}

\begin{proof}
Consider a point, ``traveling'' along the edges of the graph~$\Gamma_N$. One can describe the steps of such ``travel'' as follows: 
\begin{itemize}\itemsep 0pt
 \item when you move from an $n$-face to an $(n-1)$-face you replace one of the stars $*$ in the corresponding sequence $\tau$ by $1$ or $0$, depending on whether the star was on an odd or even place in the set of stars inside $\tau$, when you count from the left; 
 \item when you move from an $(n-1)$-face to an $n$-face you replace one of the digits in $\tau$ by the star $*$; the digit should be by $1$ or $0$, depending on whether the new star takes an even or odd place in the set of stars inside $\tau$, when you count from the left.
\end{itemize}

Assume that we have made several steps of the sort we have just described: $\tau=\tau_1\to\tau_2\to\dots\to\tau_p$. We can regard this process as a chain of changes within the same sequence $\tau$.
Suppose, that in the end all the ``stars'' in the sequence $\tau$ are at their original places, i.e. that the places, occupied by symbols $*$ in $\tau_1$ and in $\tau_p$ are the same (we don't assume anything about the rest of the sequence). In particular, the corresponding faces have the same dimension.

Let $i$ be the leftmost position in the sequence $\tau_1=\tau$, which changed at least once in the process of transformations (i.e. all the symbols on the left from $i$ in all sequences $\tau_1,\dots,\tau_p$ coincide). We can assume that the number of stars on the left of $i$ is \textit{even} (the odd case can be treated similarly). Now, if we replace a star by a digit at the $i$\/-th place, this digit should be equal to $1$; if a digit at $i$ is replaced by star, then the digit should be equal to $0$. Thus it is evident that we cannot return back to $\tau$ in this way. 
\end{proof}

Lemma  \ref{l:acycl} means that  \emph{the graph $\Gamma_N$ induces a partial order on the set of all $n$-dimensional and $(n-1)$-dimensional faces of the cube}. 

\begin{lemma}\label{l:n+1}
Let $N=n+1$. Then in an $(n+1)$-dimensional cube every absolutely incoming $(n-1)$-face precedes  every absolutely outgoing $(n-1)$-face (in the sense of the partial order given by $\Gamma_N$).
\end{lemma}

\begin{proof}
Let an absolutely incoming face $f_{n-1}$ be determined by the following system
\[
\begin{cases}
x_i=\varkappa_i,\\ x_j=\varkappa_j,
\end{cases}
\]
where
\begin{equation}\label{i<j}
i<j,
\end{equation}
and let $\varkappa_i$ and~$\varkappa_j$ be the numbers 0 or~1, chosen in accordance with the conditions of Lemma~\ref{l:abs}; that is $\varkappa_i$ and~$\varkappa_j$ are at the $i$-th and the $j$-th places in the corresponding sequence $\tau_{n-1}$, and all the other symbols in $\tau_{n-1}$ are $*$.

Similarly, consider the system, which determines an absolutely outgoing face~$h_{n-1}$:
\[
\begin{cases}
x_k=\lambda_k,\\ x_l=\lambda_l,
\end{cases}
\]
where
\begin{equation}\label{k<l}
k<l,
\end{equation}
and $\lambda_k$ and~$\lambda_l$ are 0 or~1, depending on their position (see Lemma~\ref{l:abs}).

We are going to find an ``intermediate'' face $g_{n-1}$, such that $f_{n-1}\prec g_{n-1}\prec h_{n-1}$ in the partial order, determined by $\Gamma_N$. Moreover, the face we shall choose is such, that the corresponding path from $f_{n-1}$ to $h_{n-1}$ in $\Gamma_N$ passes through only two $n$-faces and through only one intermediate $(n-1)$-face, $g_{n-1}$.
To do this, observe that since we have $i<j$ and $k<l$, we see that either $i<j$ or $k<j$. Now, if
$i<l$, then $g_{n-1}$ is given by the system
\[
\begin{cases}
x_i=\varkappa_i,\\ x_l=\lambda_l,
\end{cases}
\]
and if $k<j$, it is determined by
\[
\begin{cases}
x_k=\lambda_k,\\ x_j=\varkappa_j.
\end{cases}
\]
The fact that the conditions, mentioned above hold, can be checked by direct inspection.
\end{proof}

It follows from Lemma~\ref{l:acycl} that one can always extend an arbitrary coloring of absolutely incoming $(n-1)$-faces of a cube to a coloring of all its $(n-1)$-faces, using the operator $R$ on $n$-faces (although this continuation can be contradictory in the sense, that the same $(n-1)$-face can be painted in several different colors). Indeed, starting from any $(n-1)$-face of $I^N$ and moving against the arrows of $\Gamma_N$ we can always get to an absolutely incoming face of the big cube, since the number of $(n-1)$-faces is finite, and there are no cycles. More accurately, a face $f^0_{n-1}$ cannot be painted in this way if and only if every $n$-face $g^1_n$ for which $f^0_{n-1}$ is outgoing, contains an incoming $(n-1)$-face $f^1_{n-1}$, which cannot be painted either. We can choose such face $f^1_{n-1}$ and repeat the same reasoning for it. This process can stop only if at some step $k$ we get to a face $f^k_{n-1}$, which is not outgoing for any $n$-face of the cube. Thus this face is absolutely incoming, and hence it has been painted by our assumption. Reverting this process we find a way to color $f^0_{n-1}$. It is clear that this color can depend only on the sequence $f^0_{n-1},f^1_{n-1},\dots,f^k_{n-1}$ and is otherwise uniquely defined.

Let us now prove by induction on the given ordering of $(n-1)$-faces, that the procedure, described above, does not bring up contradictions, i.e. colors of $(n-1)$-faces do not depend on the choice of sequence $f^0_{n-1},f^1_{n-1},\dots,f^k_{n-1}$. The base of induction corresponds to absolutely incoming faces, which are all colored by assumption. To prove the inductive step, consider an $(n-1)$-face $h_{n-1}$ for which all preceding $(n-1)$-faces are colored. The face $h_{n-1}$ is outgoing for several $n$-faces. We shall show, that the $R$-operator on all these faces induce the same color on $h_{n-1}$.

It is enough to do this for any two $n$-faces containing $h_{n-1}$, so let $f_n\supset h_{n-1}$ and $g_n\supset h_{n-1}$. Then $f_n$ and $g_n$ determine a unique $(n+1)$-face $I^{n+1}\subseteq I^N$, such that $I^{n+1}\supset f_n,\ I^{n+1}\supset g_n$. The face $h_{n-1}$ is an absolutely outgoing one with respect to $I^{n+1}$. On the other hand, all absolutely incoming faces of $I^{n+1}$ precede $h_{n-1}$ (see Lemma~\ref{l:n+1}), and hence are consistently colored by inductive assumption (in fact, all the faces, preceding $h_{n-1}$ in $I^{n+1}$ are colored). Hence, by the $n$-simplex equation, the color of $h_n$, induced from $f_n$ and $g_n$ are the same.
\end{proof}

\section{Homologies and cohomologies of $n$-simplex relations}\label{s:hom}

Let $(X,R)$ be a solution of the set-theoretic $n$-simplex equation ($\mbox{STS}_n$ equation for short); here $X$ is a set, equipped with an operation $R:X^{\times n}\to X^{\times n}$, for which holds the equality \eqref{eq:n-simp}. In this section we define a new (co)homology theory, which we call \textit{the $n$\/-simplicial (co)homology}, of which a particular case is \textit{tetrahedral (co)homology}. Our construction is a direct generalization of the Yang-Baxter (co)homology theory, see \cite{CES}.

\subsection{Definition of $n$-simplicial (co)homology}
\label{ss:dfns}

As we have explained (see section \ref{ss:ur} for details), the geometric meaning of the set-theoretic $n$-simplex equation can be expressed in the terms of colorings of $(n-1)$-dimensional faces of a $(n+1)$-cube. Coupled with the construction of Yang-Baxter (co)homology (see \cite{CES}), this brings us to the following definitions.

We fix the solution $(X,R)$ of the set-theoretic $n$-simplex equation. Let $C^{n-1}(I^N)$ denote the set of all $(n-1)$\/-faces of the $N$\/-cube, $N\ge n-1$. As before, we shall say, that a \textit{(permitted) $(n-1)$-coloring of an $N$\/-cube in $X$ colors} is a map $c:C^{n-1}(I^N)\to X$, such that for every $n$\/-face $f_n\subset I^N$, colors of its outgoing $(n-1)$-faces are equal to the values of $R$ on the colors of incoming $n-1$\/-faces; i.e. if $(C^{out}_1(f_n),\dots,C^{out}_n(f_n))$ are incoming and $(C^{in}_1(f_n),\dots,C^{in}_n(f_n))$ are outgoing faces of the $n$-face $f_n$, then
$$
(c(C^{out}_{1}(f_n)),\dots,\,c(C^{out}_{n}(f_n)))=R(c(C^{in}_{1}(f_n)),\dots\,c(C^{in}_{n}(f_n))).
$$
We shall denote the set of all such colorings by $C^n(N,\,X)$.

Now we can define \emph{the complex of $n$\/-simplicial homology} of $(X,R)$ with coefficients in a ring~$\Bbbk$. In what follows we shall usually take $\Bbbk$ to be the field of real or complex numbers, and when no contradiction can occur, drop $\Bbbk$ from our notation.

\begin{dfn}
\label{dfn:n-simcx}
The complex of $n$\/-simplicial homology of $(X,R)$ with coefficients in $\Bbbk$ is the graded vector space
$$
C_*((X,R),n)=\bigoplus_{N\ge n-1}\Bbbk\cdot C^n(N,\,X),
$$
where $\Bbbk\cdot C_N(X,n)$ is a free $\Bbbk$\/-module, spanned by the set of all permitted $(n-1)$\/-colorings of the cube $I^N$. The differential $d_N:C_N((X,R),n)\to C_{N-1}((X,R),n)$ of $C_*((X,R),n)$ is given by the formula
$$
d_n(c)=\sum_{k=1}^n\left(d^f_kc-d^r_kc\right),
$$
where $d^f_kc$ (respectively $d^r_kc$) denotes the restriction of coloring $c$ on the $k$\/-th front (respectively, rear) $N-1$\/-dimensional face of the cube $I^N$, which we identify with $I^{N-1}$ in an evident way. 
\end{dfn}

Recall that in topology one says that a face $(*\;*\;\dots\;0\;*\dots\;*)$ is \textit{a front face}, and a face $(*\;*\;\dots\;1\;*\dots\;*)$ is \textit{a rear face} of the cube $I^N$. The equality $d_{N-1}d_N=0$ is an easy consequence of the fact, that result of restriction of a coloring from a face of $I^N$ to a subface does not depend on the order in which the restrictions are made.

\begin{dfn}
\label{dfn:n-simhom}
The \emph{$n$\/-simplicial homology of $(X,R)$} is the homology of the complex $C_*((X,R),n)$; we shall denote this homology by $H_*((X,R),n;\,\Bbbk)$ (or $H_*(X,n)$, dropping the $R$-operator and coefficients from notation).

Dually, \emph{$n$\/-simplicial cohomology of $(X,R)$} is the homology of dual cochain complex~$C^*((X,R),n)$:
\[
C^*((X,R),n)=\bigoplus_{N\ge n-1}C^N((X,R),n),
\]
where
\[
C^N((X,R),n)=\mathrm{Hom}\,(C_N((X,R),n),\,\Bbbk),
\]
with dual differential. It will be denoted by $H^*((X,R),n;\,\Bbbk)$, or simply $H^*(X,n)$.
\end{dfn}

In the case $n=2$ the $n$\/-simplex equation turns into the (set-theoretic) Yang-Baxter equation. In this case the definition we give reproduce the definition of \textit{Yang-Baxter (co)homology} given in \cite{CES}. If $n=3$ we shall call the corresponding homology and cohomology \textit{tetrahedral}.

The problem to compute $n$\/-simplicial (co)homology in a generic case seems to be quite a non-trivial one (c.f. the formulas for the differential in low dimensions, see section \ref{ss:exam}). The following observation can be of some use for this: the complex $C_*((X,R),n)$, defined here, is a subcomplex inside the chain complex $\tilde C_*(X)$, spanned by all possible (not necessary permitted) $(n-1)$-colorings of $I^N$ into $X$ colors: 
$$
\tilde C_N(X)=\Bbbk\cdot X^{C^{n-1}(I^N)}.
$$
The differential $\tilde d$ is given by the same formula as above. One can say, that $C_*((X,R),n)$ is  cut from $\tilde C_*(X)$ by the $n$-simplex equation.
The homology of $\tilde C_*(X)$ on the other hand is easy to compute: choose a ``zeroth'' color  $x_0\in X$;  the operation, which sends an $(n-1)$-coloring $c$ of $I^N$ to the $(n-1)$-coloring $c'$ of $I^{N+1}$, given by ascribing the color $x_0$ to all ``extra'' $(n-1)$-faces, is a chain homotopy between the identity and the ``smothering'' operation, which sends all the colors to $x_0$. Thus the homology of $\tilde C_*(X)$ are equal to the homology of the complex $\tilde C_*(x_0)$, spanned by the colorings, in which all colors coincide with $x_0$. This homology is evidently equal to $\Bbbk$ in all dimensions.

Let us make one more remark: the complex $\tilde C_*(x_0)$ coincides with the well-known topological object: the unreduced singular cubical complex of the point. In Topology this complex is usually \textit{normalized} (see \cite{HiWi}) so that its homology would coincide with the usual simplicial homology of the point. In the terms of the quasi-isomorphic complex $\tilde C_*(X)$, the normalization consists of passing to the quotient complex $\tilde C^\sharp_*(X)=\tilde C_*(X)/\check C_*(X)$, where $\check C_*(X)$ is the subcomplex, spanned by the colorings of $I^N$ with a pair of identically-colored $(N-1)$\/-faces. In a similar way we can consider \textit{normalized $n$\/-simplicial homology}:

\begin{dfn}
\label{dfn:n-simhomnorm}
\emph{Normalized $n$\/-simplicial homology} is the homology of the normalized chain complex
$$
C^\sharp_*(X,n)=C_*(X,n)/C_*(X,n)\bigcap \check C_*(X);
$$
we shall denote it by $H_*^\sharp(X,n)$. Dually one can define the \emph{normalized $n$\/-simplicial cohomology}, denoted by $H^*_\sharp(X,n)$.
\end{dfn}

\subsection{Low-dimensional examples}
\label{ss:exam}

It follows from Theorem~\ref{thm:vse} that every permitted $(n-1)$-coloring of $I^N$ is uniquely determined by colors of the absolutely incoming $(n-1)$\/-faces. This can be rephrased as follows:

\begin{thm}
\label{prop:base}
The space $C_N(X,n)$ has a base, indexed by the colorings of the set of absolutely incoming $(n-1)$-faces into $X$ colors, i.e. by the set $X^{\times{\binom{N}{n-1}}}$. In particular, if $X$ is finite, then
$$
\mathrm{dim}\,C_N(X,n)=|X|^{\binom{N}{n-1}}.
$$
\qed
\end{thm}

In a generic situation, the differential of $C_*(X,n)$ is expressed by a rather complicated formula in the terms of the base, mentioned in Theorem~\ref{prop:base}. We shall confine ourselves to giving the formulas in low dimensions and for $n=3$, i.e. for tetrahedral homology and cohomology. Formulas for other $n$ are similar.

\bigskip
\noindent\textbf{$\boldsymbol{N=2,3}$. } In dimensions $2$ and $3$ the complex $C_*(X,2)$ is very simple: 
$$
\begin{aligned}
C_2(X,2)&=\Bbbk\cdot X,\\
C_3(X,2)&=\Bbbk\cdot X^{\times3}.
\end{aligned}
$$
Let $a,\ b$ and $c$ be the colors of absolutely incoming $2$-faces of a $3$-cube; we shall denote the corresponding $2$-coloring of a three-dimensional cube by $(a,b,c)$. Similarly $(a)$ will mean the coloring of a square by the color $a$. In this notation the differential is
$$
d_3((a,b,c))=(a)+(b)+(c)-(R_1(a,b,c))-(R_2(a,b,c))-(R_3(a,b,c)).
$$

\medskip
\noindent\textbf{$\boldsymbol{N=3,4}$.} By Theorem~\ref{prop:base},
$$
C_4(X,2)=\Bbbk\cdot X^6,\ \mbox{since}\ {\binom{4}{2}}=6.
$$
As before we shall denote by $(a_1,a_2,a_3,a_4,a_5,a_6)$ the $2$\/-coloring of $I^4$, corresponding to the colors $a_1,a_2,a_3,a_4,a_5$ and $a_6$ on absolutely incoming faces. Then the differential $d_4:C_4(X,2)\to C_3(X,2)$ is given by (c.f. figure \ref{f:t}):
$$
\begin{aligned}
d_4&(a_1,a_2,a_3,a_4,a_5,a_6)=(a_1,a_2,a_3)+(a_3,a_5,a_6)\\
&-(R_1(a_1,R_2(a_2,a_4,R_3(a_3,a_5,a_6)),R_2(a_3,a_5,a_6)),\\
&\qquad\qquad\qquad\qquad\qquad\qquad\qquad\qquad R_1(a_2,a_4,R_3(a_3,a_5,a_6)),R_1(a_3,a_5,a_6))\\
&-(R_3(a_1,a_2,a_3),R_3(R_1(a_1,a_2,a_3),a_4,a_5),\\
&\qquad\qquad\qquad\qquad\qquad\qquad\quad R_3(R_2(a_1,a_2,a_3),R_2(R_1(a_1,a_2,a_3),a_4,a_5),a_6))\\
                                                         &+(a_1,R_2(a_2,a_4,R_3(a_3,a_5,a_6)),R_2(a_3,a_5,a_6))-(a_2,a_4,R_3(a_3,a_5,a_6))\\
                                                         &+(R_2(a_1,a_2,a_3),R_2(R_1(a_1,a_2,a_3),a_4,a_5),a_6)-(R_1(a_1,a_2,a_3),a_4,a_5).
\end{aligned}
$$
Note that dualization of this formula yields the following description of tetrahedral $3$-cocycles: they are the maps $f:X^{\times 3}\to\Bbbk$, vanishing on the image of $d_4$:
$$
\begin{aligned}
f&(a_1,a_2,a_3)-f(R_1(a_1,a_2,a_3),a_4,a_5)\\
 &+f(R_2(a_1,a_2,a_3),R_2(R_1(a_1,a_2,a_3),a_4,a_5),a_6)\\
 &-f(R_3(a_1,a_2,a_3),R_3(R_1(a_1,a_2,a_3),a_4,a_5),\\
 &\qquad\qquad\qquad\qquad\qquad\qquad\qquad R_3(R_2(a_1,a_2,a_3),R_2(R_1(a_1,a_2,a_3),a_4,a_5),a_6))\\
 &=-f(a_3,a_5,a_6)+f(a_2,a_4,R_3(a_3,a_5,a_6))\\
 &\ -f(a_1,R_2(a_2,a_4,R_3(a_3,a_5,a_6)),R_2(a_3,a_5,a_6))\\
 &\ +f(R_1(a_1,R_2(a_2,a_4,R_3(a_3,a_5,a_6)),R_2(a_3,a_5,a_6)),\\
 &\qquad\qquad\qquad\qquad\qquad\qquad\qquad\qquad R_1(a_2,a_4,R_3(a_3,a_5,a_6)),R_1(a_3,a_5,a_6)).
\end{aligned}
$$
In the notation of formula \eqref{R}, this equality shrinks to
$$
\begin{aligned}
f&(a_1,a_2,a_3)-f(a_1',a_4,a_5)+f(a_2',a_4',a_6)-f(a_3',a_5',a_6')=\\
 &=-f(a_3,a_5,a_6)+f(a_2,a_4,a_6')-f(a_1,a_4',a_5')+f(a_1',a_2',a_3').
\end{aligned}
$$
Below we shall make use of the following multiplicative form of this equation: put  
$$
\varphi(a,b,c)=e^{(-1)^\varepsilon f(a,b,c)},
$$
where $\varepsilon=(-1)^{\#'}$ ($\#'$ is the number of ``primed'' arguments on the left and $1$ plus this number on the right), then $\varphi$ is cocycle iff
\begin{equation}\label{cocycle}
\begin{aligned}
\varphi&(a_1,a_2,a_3)\varphi(a_1',a_4,a_5)\varphi(a_2',a_4',a_6)\varphi(a_3',a_5',a_6')\\
 &=\varphi(a_3,a_5,a_6)\varphi(a_2,a_4,a_6')\varphi(a_1,a_4',a_5')\varphi(a_1',a_2',a_3').
\end{aligned}
\end{equation}
In the same notation, $\varphi$ is cohomologous to $0$, iff 
\begin{equation}\label{coboundary}
\varphi(a,b,c)=\frac{\psi(a)\psi(b)\psi(c)}{\psi(a')\psi(b')\psi(c')}
\end{equation}
for some $\psi:X\to\Bbbk$.

The introduced notion of cocycle looks productive: there do exist nontrivial cocycles. This will be shown in Subsection~\ref{ss:xmpl}, see Lemma~\ref{lem:nontrivial}. Right now, we observe that both the 3-cocycle relation~\eqref{cocycle} and 3-coboundary relation~\eqref{coboundary} make sense also for the FTE with $R$ being a rational mapping such as~\eqref{el}. So we can generalize (at least) the notions of 3-cocycle and 3-coboundary onto the case of FTE, and the following lemma shows that, indeed, there exist interesting cocycles in the electric case.

\begin{lemma}\label{lem_cocycle}
Consider the electric solution~\eqref{el} to tetrahedron equation $\Phi\colon\; (a_1,a_2,a_3) \mapsto (a_1',a_2',a_3')$. For this case, the following expressions, as well as the product of them raised in any integer degrees, are 3-cocycles of the tetrahedral complex:
\begin{equation}\label{el-cocycles}
 \begin{array}{rcl}
  c_1(a_1,a_2,a_3, a_1',a_2',a_3')&=&a_2, \\[.8ex]
  c_2(a_1,a_2,a_3, a_1',a_2',a_3')&=&a_2'.
 \end{array}
\end{equation}
\end{lemma}

\begin{proof}
Direct calculation.
\end{proof}

\section{Solutions to quantum tetrahedron equations obtained from nontrivial cocycles}\label{s:examples}

The $n$-simplicial homologies and cohomologies introduced above are expected to find many applications in topology, theory of dynamical systems, and other parts of mathematics. Below we present an example of how the notion of cocycle can be applied to obtaining new solutions to the QTE (quantum tetrahedron equation).

\subsection{Cocycles and solutions to QTE}\label{ss:QTE-cocycles}

If R is a solution to STTE on a finite set~$X$, then there is a canonical construction of a solution to QTE on $V=\Bbbk[X]$ --- the space of functions on~$X$ taking values in a field~$\Bbbk$. In this space, we introduce the basis $\{e_x\}_{x\in X}$. Consider the linear mapping in $V\times V\times V$ defined as follows:
\begin{equation}\label{PhiV}
\Phi(e_x\otimes e_y\otimes e_z):=e_{x'}\otimes e_{y'}\otimes e_{z'}\qquad\text{where}\qquad (x',y',z')=R(x,y,z).
\end{equation}

We have then the following simple lemma:

\begin{lemma}\label{l:V}
If $R$ is a solution to STTE, then $\Phi$ is a solution to QTE.
\end{lemma}

\begin{proof}
First, as $\Phi$ acts in the very same manner as $R$ on \emph{basis} vectors $e_x\otimes e_y\otimes e_z \in V\times V\times V$, both sides of QTE also act in the same way on all basis vectors in~$V^{\times 6}$. Second, this equalness is extended onto the whole~$V^{\times 6}$ by linearity.
\end{proof}

Solutions like this~$\Phi$ are called \emph{permutation-type} solutions, see~\cite{Hie}.

\begin{dfn}
Two solutions $\Phi$ and~$\Phi'$ to the quantum constant tetrahedron equation are called \emph{equivalent} if they can be obtained from one another by a conjugation with the tensor product of three copies of some (invertible) operator~$A$, each copy acting in its own copy of~$V$:
\begin{equation}\label{ABC}
\Phi' = (A\otimes A\otimes A)\,\Phi\,(A\otimes A\otimes A)^{-1}.
\end{equation}
\end{dfn}

We are going to generalize the permutation-type solutions as follows:

\begin{thm}
{\rm(a)}~If $R$ is a solution to STTE, and $\varphi$ is a multiplicative cocycle (i.e., \eqref{cocycle} holds), then the operator determined by
\begin{equation}\label{PhiVeta}
\Phi_{\varphi}(e_x\otimes e_y\otimes e_z) = \varphi(x,y,z) e_{x'}\otimes e_{y'}\otimes e_{z'}
\end{equation}
also is a solution to QTE. 

{\rm(b)}~If $\varphi$ and~$\varphi'$ are two cohomologous cocycles, then the corresponding operators $\Phi_{\varphi}$ and~$\Phi_{\varphi'}$ are gauge equivalent and, moreover, (matrix of) operator~$A$ in~\eqref{ABC} is diagonal.
\end{thm}

\begin{proof}
{\rm(a)}~Again, like in Lemma~\ref{l:V}, both sides of QTE act in the same way on all basis vectors in~$V^{\times 6}$. This time, the resulting basis vectors acquire (in contrast with situation in Lemma~\ref{l:V}) some nontrivial multipliers, but these coincide in both sides, due to the equalities~\eqref{cocycle}.

{\rm(b)}~If $\varphi$ and~$\varphi'$ are cohomologous, then, according to~\eqref{coboundary},
\begin{equation*}
\frac{\varphi(x,y,z)}{\varphi'(x,y,z)}=\frac{\psi(x)\psi(y)\psi(z)}{\psi(x')\psi(y')\psi(z')}
\end{equation*}
Replacing $\varphi$ with $\varphi'$ in~\eqref{PhiVeta}, we get thus:
\begin{equation}\label{gequiv}
\Phi_{\varphi'}\bigl(\psi(x)e_x\otimes \psi(y)e_y\otimes \psi(z)e_z\bigr) = \varphi(x,y,z) \psi(x')e_{x'}\otimes \psi(y')e_{y'}\otimes \psi(z')e_{z'}.
\end{equation}
And comparing now \eqref{gequiv} with~\eqref{PhiVeta}, we see that \eqref{ABC} indeed holds, with diagonal $A\colon\; e_x\mapsto \psi(x)e_x$.
\end{proof}

\subsection{Avoiding the singularities in electric solution}\label{ss:sol}

The electric solution~\eqref{el} to FTE
\begin{equation*}
\begin{aligned}
R(x,y,z)&=(x',y',z');\\
x'&=\frac {x y} {x+z+x y z},\\
y'&=x+z+x y z,\\
z'&=\frac{ yz} {x+z+x y z}
\end{aligned}
\end{equation*}
is not, generally speaking, a solution to STTE, because of possible singularities. Remarkably, there is a possibility to avoid these singularities by taking a special color set~$X$. Consider the residue ring~$\Z/p^k\Z$, where $p$ is either a prime number of the form $p=4l+1$, or $p=2$, and $k$ is an integer $\ge 2$. For such~$p$, the Legendre symbol $\left(\frac{-1}{p}\right)=1$, that is, there exists a square root of $-1$ in~$\Z/p\Z$. We fix one such square root and call it $\varepsilon\in \Z/p\Z$. Our~$X$ will be the following subset of~$\Z/p^k\Z$:
\begin{equation}\label{X}
X=\{x \in \Z/p^k\Z\,\colon\;\; x=\varepsilon \mod p\}.
\end{equation}

\begin{lemma}
The electric mapping~\eqref{el} restricts correctly to a bijection of set $X\times X\times X$ onto itself.
\end{lemma}

\begin{proof}
First, we note that
\begin{equation*}
y' \mod p = \varepsilon+\varepsilon + \varepsilon^3=\varepsilon,
\end{equation*}
as required. As $y'$ also coincides with the denominator in the expressions for both $x'$ and~$z'$, and clearly
\begin{equation}\label{vareps}
\varepsilon \ne 0 \mod p^k,
\end{equation}
the division in those expressions goes without trouble. We then see that $x' \mod p = z' \mod p = \varepsilon$ as well.
\end{proof}

Recall that we agreed to take field $\mathbb R$ or~$\mathbb C$ as our ring~$\Bbbk$ of coefficients (Subsection~\ref{ss:dfns}, paragraph before Definition~\ref{dfn:n-simcx}). What remains is to make from our expressions~\eqref{el-cocycles} multiplicative $\Bbbk$-valued cocycles. To this end, we can observe that all elements in~$X$ are \emph{invertible} (according to \eqref{X} and~\eqref{vareps}), and apply to expressions~\eqref{el-cocycles} any homomorphism
\begin{equation}\label{eta}
\eta\colon\;\; (\Z/p^k Z)^* \to \Bbbk^*
\end{equation}
of groups of invertible elements in our rings.

\subsection{Examples}\label{ss:xmpl}

\begin{Ex}
Consider the case~$\Z/25\Z$, and fix $\varepsilon=2$. Introduce a coordinate~$n$ on the set~$X$ so that $x=2+5n$, where $x\in X$ and $n$ takes values $n=0,\ldots,4$, identified with elements in~$\Z/5\Z$. To be more exact, below in~\eqref{xyz}
\begin{equation}\label{xn}
x=2+5n_1,\quad y=2+5n_2,\quad z=2+5n_3.
\end{equation}
In these notations, the electric transformation takes the form
\begin{equation}\label{xyz}
R(n_1,n_2,n_3)=(n_1+2n_2-2,\,2-n_2,\,n_3+2n_2-2).
\end{equation}
The multiplicative group~$(\Z/25\Z)^*$ is a cyclic group of order~20. So, if we choose $\Bbbk=\mathbb C$, then there are 20 possible homomorphisms~$\eta$~\eqref{eta}. The expressions~\eqref{el-cocycles} yield the following cocycles:
\begin{equation}\label{Z25}
\eta(y)=\eta(2+5n_2)\quad\text{and}\quad \eta(y')=\eta\bigl(2+5(2-n_2)\bigr).
\end{equation}
\end{Ex}

\begin{lemma}\label{lem:nontrivial}
There are nontrivial --- not being a coboundary --- cocycles among those in~\eqref{Z25}.
\end{lemma}

\begin{proof}
Take $n_2=1$ and arbitrary $n_1$ and~$n_3$ in~\eqref{xyz}. Then \eqref{xyz} gives $(n'_1,n'_2,n'_3) = (n_1,n_2,n_3)$, hence also $(x',y',z')=(x,y,z)$, and formula~\eqref{coboundary} gives value~$1$ if our cocycle is a coboundary. But, according to the second formula in~\eqref{xn}, $y=y'=7$, and there obviously exist homomorphisms~$\eta$ such that $\eta(7)\ne 1$.
\end{proof}

\begin{Ex}
A different example of the construction in Subsection~\ref{ss:sol} arises if we take the set of odd elements in the ring $\Z/2^k\Z$ as our finite set~$X$, and restrict the electric solution~\eqref{el} onto this~$X$. Consider in more detail the case $k=3$. In this case, $X$ can be identified with $\Z/4\Z$ using the mapping $\phi\colon\;\Z/4\Z\rightarrow\Z/8\Z$ given by
\begin{equation*}
\phi\colon\;n\mapsto 2n+1.
\end{equation*}
In this parameterization, the solution $R\colon\;(n_1,n_2,n_3)\mapsto(n_1',n_2',n_3')$ to STTE can be written as follows:
\begin{eqnarray*}
n_1'&=&n_1+1+2(n_2+n_3+n_1 n_2 + n_1 n_3+n_2 n_3); \\
n_2'&=&n_2+1+2(n_1+n_3+ n_1 n_2 + n_1 n_3+n_2 n_3); \\
n_3'&=&n_3+1+2(n_1+n_2+n_1 n_2 + n_1 n_3+n_2 n_3).
\end{eqnarray*}
The cocycles are, accordingly,
\begin{equation}\label{eta8}
\eta(2n_2+1)\quad\text{and}\quad\eta(2n'_2+1),
\end{equation}
where, as the group $(\Z/8\Z)^*$ is isomorphic to $\Z_2\times\Z_2$, there are this time four possible~$\eta$ for either $\Bbbk=\mathbb R$ or~$\mathbb C$. 
We leave the question of (possible) nontriviality of cocycles~\eqref{eta8}, as well as their analogues for arbitrary ring~$\Z/2^k\Z$, for further research.
\end{Ex}

\subsection*{Acknowledgements}

The work of G.~S.\ was partially supported by the RFBR grant 14-01-00007-a and the grant NSh-1410.2012.1. The work of D.~T.\ was partially supported by the RFBR grant 14-01-00012, by the RFBR-CNRS grant 11-01-93105-a, by the grant of the Dynasty foundation, by the grant RFBR-OFI-I2 13-01-12401 and by the grant for support of the scientific school 4833.2014.1.


\begin{thebibliography}{99}

\bibitem{Dri}
V. G. Drinfeld, {\em On some unsolved problems in quantum group theory}. In Quantum Groups
(Leningrad, 1990), Lecture Notes in Math. 1510, Springer-Verlag, Berlin, 1992, 1--8.

\bibitem{MR}
N. Metropolis and Gian-Carlo Rota,
{\em Combinatorial Structure of the Faces of the $n$-Cube},
SIAM Journal on Applied Mathematics, Vol. 35, No. 4 (Dec., 1978), pp. 689--694.

\bibitem{JSC}
J. S. Carter, D. Jelsovsky, L. Langford, S. Kamada, M. Saito, {\em Quandle cohomology and
state-sum invariants of knotted curves and surfaces,} Trans. Amer. Math. Soc. {\bf 355} (2003), no.
10, 3947-3989.

\bibitem{Fish}
J. E. Fischer, Jr. {\em 2-Categories and 2-knots.} Duke Mathematical Journal, 1994 Vol. {\bf 75}, No. 2, pp. 493-596.

\bibitem{Hie}
J. Hietarinta,
{\em Permutation-type solutions to the Yang-Baxter and other n-simplex equations.}
J. Phys. A: Math. Gen., 1997, Vol. {\bf 30}, 4757, arXiv:q-alg/9702006.

\bibitem{Ros}
D. Roseman, {\em Reidemeister-type moves for surfaces in four-dimensional space.} Knot theory, Banach center publications, Vol {\bf 42}, Institute of mathematics, Polish academy of sciences. Warszawa 1998.

\bibitem{YBE1}
C. N. Yang, {\em Some exact results for the many-body problem in one dimension with repulsive delta-function interaction}, Phys. Rev. Letters {\bf 19} (1967), 1312--1315.

\bibitem{YBE2}
R. J. Baxter, {\em Exactly Solved Models in Statistical Mechanics}, Academic Press, 1982.


\bibitem{ZTE} 
A. Zamolodchikov, {\em Tetrahedron equations and the relativistic $S$-matrix of straight-strings in $2+1$-Dimensions}, Comm. Math. Phys., Volume 79, Issue 4 (1981), pp.~489--505.

\bibitem{CKS}
S. Carter, S. Kamada, M. Saito, {\em Surfaces in 4-space}, Encyclopaedia of mathematical sciences, Volume 142. Springer-Verlag 2004.

\bibitem{FTE}
R. M. Kashaev, I. G. Korepanov, S. M. Sergeev.
{\em Functional tetrahedron equation}. Theor. Math. Phys., Volume 117, Issue 3 (1998), pp.~1402--1413.

\bibitem{CES} J. S. Carter, M. Elhamdadi, M. Saito, {\em Homology Theory for the Set-Theoretic Yang-Baxter Equation and Knot Invariants from Generalizations of Quandles}, Fund. Math. {\bf 184} (2004), 31--54.

\bibitem{HiWi} P. J. Hilton, S. Wylie, {\em Homology theory: An introduction to algebraic topology}, Cambridge University Press, 1960.

\bibitem{S-et-al}
V. V. Mangazeev, V. V. Bazhanov and S. M. Sergeev, {\em An integrable 3D lattice model with positive Boltzmann weights}, J. Phys. A: Math. Theor. {\bf 46} 465206.

\bibitem{Turaev}
V. G. Turaev, {\em The Yang-Baxter equation and invariants of links}, Inventiones Mathematicae,
Volume 92, Issue 3 (1988), pp. 527--553.

\end{thebibliography}
\end{document}